\documentclass[twocolumn,prl,secnumarabic,nobibnotes,superscriptaddress,aps,10pt]{revtex4-2}

\usepackage[utf8]{inputenc}
\usepackage{CJKutf8}
\usepackage{CJK}
\usepackage{graphicx}% Include figure files
\usepackage{dcolumn}% Align table columns on decimal point
\usepackage{bm}% bold math
\usepackage{mathtools}
\usepackage{amssymb,amsmath,amsfonts,dsfont, amsthm,pifont}
\theoremstyle{definition}

\newtheorem{theorem}{Theorem}

\usepackage{multirow}
\usepackage[hidelinks]{hyperref}
\hypersetup{colorlinks=true} 
\usepackage[capitalize]{cleveref}
\usepackage{color}
\usepackage[ruled, lined, linesnumbered, commentsnumbered, longend]{algorithm2e}

\newcommand{\rev}[1]{{#1}}

\def\bra#1{\mathinner{\langle{#1}|}}
\def\ket#1{\mathinner{|{#1}\rangle}}

\newcommand{\proj}[1]{\ket{#1}\!\!\bra{#1}}

\newcommand{\Or}{\mathcal{O}}

\newcommand{\cmark}{\textcolor{green}{\ding{51}}}%
\newcommand{\xmark}{\textcolor{red}{\ding{55}}}%

\newcommand{\DeptMath}{Department of Mathematics, University of California, Berkeley, CA 94720, USA}
\newcommand{\LBLMath}{Applied Mathematics and Computational Research Division, Lawrence Berkeley National Laboratory, Berkeley, CA 94720, USA}
\newcommand{\Caltechappliedscience}{Division of Engineering and Applied Science, California Institute of Technology, Pasadena, California 91125, USA}
\newcommand{\CaltechChem}{Division of Chemistry and Chemical Engineering, California Institute of Technology, Pasadena, California 91125, USA}
\begin{document}

\begin{CJK*}{UTF8}{mj}
\title{Coupled Lindblad pseudomode theory for simulating open quantum systems}

\author{Zhen Huang}
\affiliation{\DeptMath}
\author{Gunhee Park(박건희)}
\affiliation{\Caltechappliedscience}
 
\author{Garnet Kin-Lic Chan}
\affiliation{\CaltechChem}
\author{Lin Lin}
\thanks{linlin@math.berkeley.edu}
\affiliation{\DeptMath}
\affiliation{\LBLMath}

\date{\today}

\begin{abstract}Coupled Lindblad pseudomode theory is a promising approach for simulating non-Markovian quantum dynamics on both classical and quantum platforms, with dynamics that can be realized as a quantum channel. We provide theoretical evidence that the number of coupled pseudomodes only needs to scale as $\mathrm{polylog}(T/\varepsilon)$ in the simulation time $T$ and precision $\varepsilon$. Inspired by the realization problem in control theory, we also develop a robust numerical algorithm for constructing the coupled modes that avoids the non-convex optimization required by existing approaches. We demonstrate the effectiveness of our method by computing population dynamics and absorption spectra for the spin-boson model. This work provides a significant theoretical and computational improvement to the coupled Lindblad framework, which impacts a broad range of applications from classical simulations of quantum impurity problems to quantum simulations  on near-term quantum platforms.

\end{abstract}

%\keywords{Suggested keywords}
\maketitle
\end{CJK*}

\emph{Introduction.--}
Simulating the non-Markovian dynamics of open quantum systems linearly coupled to Gaussian environments is a central challenge in quantum science~\cite{RivasHuelga2012,BreuerPetruccione2002,BerkelbachThoss2020,LeggettChakravartyDorseyetal1987}, particularly beyond the weak-coupling regime. When simulating the non-Markovian dynamics on a classical or quantum computer via an approximate representation, two key questions arise: (1) \emph{Efficiency}: Can the simulation be performed using minimal computational resources? (2) \emph{Physicality}: Does the approximate dynamics correspond to a physically realizable process?

We focus on approximation schemes that employ a finite number of auxiliary bath modes (often called the bath, environment, discrete modes, or pseudomodes in various settings). These auxiliary modes effectively capture the influence of the environment on the system dynamics. \rev{The efficiency of such schemes is characterized by the number of modes required to approximate the environment. The goal is to simulate all bounded system observables up to time \(T\) within precision \(\varepsilon\). This task is often reduced to reproducing the bath correlation functions (BCFs) with comparable precision~\cite{TamascelliSmirneHuelgaetal2018, MascherpaSmirneHuelgaetal2017, vilkoviskiy2023bound, LiuLu2025, HuangLinParketal2024}. 
The approximation scheme is considered efficient if the number of modes scales as $\text{polylog}(T/\varepsilon)$.}

\rev{Physicality requires that the combined system-bath dynamics is formulated as a quantum channel, ensuring completely positive and trace-preserving (CPTP) dynamics.} \rev{Lindbladian dynamics~\cite{Lindblad1976, Gorini1976} serves as a standard example of physical dynamics.}
The CPTP property, in turn, guarantees numerical stability when simulating the dynamics on a classical computer. Physicality is also essential for efficient implementation on a quantum computer~\cite{CW17, LiWang2023, DingLiLin2024, Lemmer2018, Kang2024,BertrandBesserveFerreroetal2024}.

\Cref{tab:compare_method} summarizes and compares existing schemes. The most widely used approach is the unitary discrete mode representation, where both the system and auxiliary modes evolve under unitary dynamics~\cite{Prior2010, DeVegaInesSchollwock2015}. Despite its wide usage, the number of modes scales linearly in the simulation time $T$~\cite{DeVegaInesSchollwock2015, WoodsCramerPlenio2015, WoodsPlenio2016}. Intuitively, this limitation arises because unitary dynamics of a finite system lack dissipation, making it incapable of accurately modeling BCFs that decay over time. 

\begin{table*}[ht]
\begin{tabular}{c|c|c|c|l}
\hline
Method & Number of & Quantum & Dissipation & Reliable numerical \\ 
 & modes & Channel & &  algorithm \\ \hline
Unitary mode~\cite{Prior2010,DeVegaInesSchollwock2015,WoodsPlenio2016}  & $\mathrm{poly}(T)\mathrm{polylog}(1/\varepsilon)$        & \cmark                                           & \xmark                                               &                                      \cmark \\ \hline
Lorentzian pseudomode~\cite{TrivediMalzCirac2021}  & $\mathrm{poly}(T / \varepsilon)$ & \cmark                                           & \cmark                                               &                                    
  ? \\ \hline
Non-Hermitian pseudomode~\cite{Lambert2019, Pleasance2020}  & $\mathrm{polylog}(T / \varepsilon)$  & \xmark                                           & \cmark                                               &  
  \cmark \\\hline
Quasi-Lindblad pseudomode~\cite{ParkHuangZhuetal2024,ThoennissVilkoviskiyAbanin2025}  & $\mathrm{polylog}(T /\varepsilon)$  & \xmark                                           & \cmark                                               &  
  \cmark \\\hline
Previous works on coupled Lindblad \cite{MascherpaSmirneSomozaEtAl2020,LednevGarciaFeist2024, Dorda2014, Dorda2015} & ?  & \cmark                                           & \cmark                                               & 
  ? \\\hline
\textbf{This work} & $\mathrm{polylog}(T / \varepsilon)$   & \cmark                                           & \cmark                                               &            
  \cmark 
\\\hline
\end{tabular}
\caption{Comparison of finite mode approximations of the environment for simulating non-Markovian dynamics. The expressions under the ``Number of modes" column indicate provable scaling for approximating the bath correlation function up to time $T$ within precision $\varepsilon$. 
}
\label{tab:compare_method}
\end{table*}

The pseudomode theory~\cite{Garraway1997, Dalton2001,
MazzolaManiscallcoPiiloetal2009, TamascelliSmirneHuelgaetal2018,
Somoza2019, MascherpaSmirneSomozaEtAl2020,Chen_2019,Li2021,Lorenzoni2024,lorenzoni2025} aims at addressing this limitation by introducing dissipation in order to model energy relaxation correctly. In Refs.~\cite{Garraway1997, Dalton2001, TamascelliSmirneHuelgaetal2018}, the environment is represented by bath modes, each subject to Lindblad dissipation. \rev{In this scheme,} the spectral density \rev{is expressed as} a sum of Lorentzians, which \rev{we refer} to as a Lorentzian pseudomode. However, a fundamental drawback is that the tail of a Lorentzian exhibits only inverse polynomial rather than exponential decay in the frequency domain.
Theoretical analyses indicate that even for smooth spectral densities, the number of pseudomodes scales as \( \text{poly}(T/\varepsilon) \)~\cite{TrivediMalzCirac2021}, though the preconstant can be smaller than that of unitary discrete modes. 

Recent advances in pseudomode methods, including non-Hermitian~\cite{Lambert2019, Pleasance2020, Cirio2023, MenczelFunoCirio2024etal, CirioLuoLiangetal2024} and quasi-Lindblad pseudomodes~\cite{ParkHuangZhuetal2024, ThoennissVilkoviskiyAbanin2025}, have significantly improved the efficiency. Under certain analytic conditions on the spectral density, the \rev{required} number of pseudomodes scales as \( \mathrm{polylog}(T/\varepsilon) \)~\cite{vilkoviskiy2023bound, ThoennissVilkoviskiyAbanin2025}.  However, these dynamics are not completely positive (CP) and cannot be realized as quantum channels. 

All the finite mode approximations discussed so far are decoupled, meaning that the modes interact with the system but not with one another. There is another class of approximation schemes that allows couplings between pseudomodes while preserving the Lindblad form~\cite{MascherpaSmirneSomozaEtAl2020, LednevGarciaFeist2024, Dorda2014, Dorda2015}, and hence the resulting dynamics \rev{is physical}. Although introduced in different contexts, we collectively refer to these as \emph{coupled Lindblad pseudomode} theory. In particular, empirical studies in Refs.~\cite{MascherpaSmirneSomozaEtAl2020, LednevGarciaFeist2024} demonstrated that only a small number of such coupled pseudomodes are sufficient to \rev{approximate the BCF accurately}. However, it remains unclear whether \rev{this provides a superior asymptotic scaling or simply a smaller prefactor.} Moreover, the construction in~\cite{MascherpaSmirneSomozaEtAl2020, LednevGarciaFeist2024, Dorda2014, Dorda2015} relies on non-convex optimization, which can be challenging in practice.

The contributions of this Letter are twofold: (1) We establish a direct connection between the coupled Lindblad and quasi-Lindblad pseudomode theory. Combined with the theoretical results of Ref.~\cite{ThoennissVilkoviskiyAbanin2025} on spectral density fitting,
\rev{this work provides a theoretical justification for the efficiency of the coupled Lindblad pseudomodes with $\mathrm{polylog}(T/\varepsilon)$ scaling.} (2) Inspired by the \emph{realization problem} in control theory~\cite{MayoAntoulas2007}, we develop a robust numerical algorithm for constructing the coupled modes that avoids the non-convex optimization used in Refs.~\cite{MascherpaSmirneSomozaEtAl2020, LednevGarciaFeist2024, Dorda2014, Dorda2015}. This can significantly simplify the \rev{construction of} coupled Lindblad pseudomodes, \rev{yielding higher accuracy while maintaining a minimal number of modes.} Taken together, this provides a firm theoretical and computational foundation for the widespread application of \rev{pseudomode-based simulations}.

 \vspace{1em}
 
 \emph{Model setup.--}
For concreteness, we consider a spin-boson Hamiltonian $\hat{H} = \hat{H}_{\text{S}} + \hat{H}_{\text{B}} + \hat{H}_{\text{SB}}$, where $\hat H_{\text B} = \int_{0}^{\infty} \omega \hat b_{\omega}^\dagger \hat b_{\omega}\mathrm d \omega$, $\hat{H}_{\text{SB}} = \hat{S}\hat{B}$, $\hat{S},\hat{B}$ are Hermitian, and $\hat{B} =  \int_{0}^{\infty} \sqrt{J(\omega)} (\hat b_{\omega} +\hat b_{\omega}^\dagger)\mathrm d\omega $. Here, $\hat{b}_\omega$ is a bosonic annihilation operator and $J(\omega)$ is the spectral density. Assuming an initially factorized state $\hat{\rho}(0) = \hat{\rho}_{\text{S}}(0) \otimes \hat{\rho}_{\text{B}}(0)$, where $\hat{\rho}_{\text{B}}(0)$ is the thermal equilibrium state of $\hat H_{\text B}$, our goal is to compute the dynamics of the system-reduced density operator $\hat{\rho}_{\text{S}}(t) = \operatorname{tr}_{\text{B}}\left({\mathrm{e}}^{-\mathrm i\hat{H}t} \hat{\rho}(0) {\mathrm{e}}^{\mathrm i\hat{H}t}\right)$, where the influence of environment is fully captured by the BCF,
\begin{equation}
    C(t) =\operatorname{tr}\left(\hat{B}(t) \hat{B}(0)  \hat{\rho}_{\text B}(0) \right),
    \label{eq:BCF}
\end{equation}
with $\hat{B}(t) = {\mathrm{e}}^{\mathrm i\hat{H}_{\text B}t} \hat{B} {\mathrm{e}}^{-\mathrm i\hat{H}_{\text B}t}$.

\vspace{1em} 
\emph{Pseudomode theory.--}
Pseudomode theory introduces a finite number of auxiliary modes (denoted as A). If the pseudomode BCF matches the original BCF for all \( t \in [0, T] \), then the corresponding pseudomode dynamics exactly reproduces the reduced system density operator \( \hat{\rho}_{\text{S}}(t) \) up to time \( T \)~\cite{TamascelliSmirneHuelgaetal2018}.

The dynamics of the density operator (denoted as $\hat\rho_{\text{SA}}^{\text c}(t)$) in the coupled Lindblad pseudomode is:
\begin{equation}
\begin{aligned}
    &\frac{\mathrm d}{\mathrm dt} \hat\rho_{\text{SA}}^{\text{c}} =  -\mathrm i[\hat H_{\text S}+\hat H_{\text A}+\hat H_{\text{SA}},\hat\rho_{\text{SA}}^{\text{c}}] + \boldsymbol{D}_{\text A}(\hat\rho_{\text{SA}}^{\text{c}}),\quad 
    \\&\hat H_{\text A} = \sum_{k,l = 1}^N H_{kl} \hat b_k^\dagger \hat b_l,\quad 
    \hat H_{\text{SA}} = \hat S \hat{A},
    \\  
    &\boldsymbol{D}_{\text A}(\bullet)  = \sum_{k,l=1}^N \Gamma_{kl} \left(2 \hat b_l \bullet \hat b_k^\dagger  - \left\{\hat b_k^\dagger \hat b_l, \bullet\right\} \right),
\end{aligned}
\label{eq:coupled_Lindblad}
\end{equation}
with the bath initially in the vacuum state, i.e., $\hat{\rho}_{\text A}(0) = \proj{\boldsymbol{0}}$. Here, $H = H^\dagger$, $\Gamma\succeq 0$, and $\hat{A} = \sum_k g_k \hat b_k^\dagger  + \overline{g_k} \hat b_k $, where we denote a positive semidefinite (definite) matrix $M$ as $M \succeq 0$ ($M\succ 0$). The conditions on $H$ and $\Gamma$ ensure that the dynamics of $\hat{\rho}_\text{SA}^\text{c}$ are CPTP, i.e. physical. The term \emph{coupled} modes indicates that both $H$ and $\Gamma$ can be dense matrices; the off-diagonal elements mediate couplings between modes. In contrast, the Lorentzian pseudomode assumes $H$ and $\Gamma$ are diagonal, and the modes are \emph{decoupled}. The BCF of the coupled Lindblad pseudomode, $C^{\text{c}}(t) = \operatorname{tr}\left(\hat{A}(t) \hat{A}(0)  \hat{\rho}_{\text A}(0) \right) $, is given by~\cite{HuangLinParketal2024,LednevGarciaFeist2024}: 
\begin{equation}
    C^{\text{c}}(t) =  g^\dagger \mathrm e^{-\mathrm i Kt} g,\quad K  = H - \mathrm i \Gamma.\label{eq:matching_Lindblad} 
\end{equation}

Another pseudomode theory in comparison is the quasi-Lindblad theory~\cite{ParkHuangZhuetal2024}. It includes additional system-bath dissipation to Eq.~\eqref{eq:coupled_Lindblad},
\begin{equation}
    \boldsymbol{D}_{\text{SA}}(\bullet) = \hat{L}^q \bullet \hat{S} + \hat{S} \bullet \hat{L}^{q\dagger} - \frac{1}{2} \{ \hat{S} (\hat{L}^q + \hat{L}^{q\dagger}) , \bullet  \},
\end{equation}
where $\hat{L}^q = \sum_k 2\alpha_k \hat{b}_k$. With $l_k, r_k = g_k \pm \rm{i} \alpha_k$ and assuming diagonal $H$ and $\Gamma$ ($H_{kk} = \omega_k, \Gamma_{kk}=\gamma_k$), the corresponding BCF is given by \cite{ParkHuangZhuetal2024}:
\begin{equation}
    C^{\text q}(t) =\sum_{k=1}^N \overline{l_k} r_k \mathrm e^{(-\mathrm i\omega_k-\gamma_k)t}
    = l^\dagger \mathrm e^{-\mathrm i \Lambda t} r,
    \label{eq:matching_quasi}
\end{equation}
where $\Lambda = \text{diag}(\omega_k - \mathrm i \gamma_k)$ is  diagonal.  The Lorentzian pseudomode is recovered by setting $l_k=r_k$ (namely, $\alpha_k=0$) so that each exponential term in  Eq.~\eqref{eq:matching_quasi} has positive weights, but in general, the weights $\overline{l_k} r_k$ are complex-valued.  Several algorithms \cite{ParkHuangZhuetal2024, ThoennissVilkoviskiyAbanin2025, XuYanShietal2022, Takahashi2024, ZhangErpenbeckAndreetal2025} could be used to accurately fit the BCF in the form of Eq.~\eqref{eq:matching_quasi} with complex weights. Under certain analyticity assumptions, the number of modes scales as $\mathrm{polylog}(T/\varepsilon)$~\cite{ThoennissVilkoviskiyAbanin2025}, which is significantly more efficient than the unitary and Lorentzian pseudomode approaches~\cite{TrivediMalzCirac2021}.

The system-bath dissipation $\boldsymbol{D}_{\text{SA}}$ in general breaks the CP condition  when there is no dissipation acting on the system. Violating the CP condition can induce instabilities in the quasi-Lindblad dynamics, posing challenges for classical simulation~\footnote{There is a subtle concept called the Hamiltonian-induced stability~\cite{ParkHuangZhuetal2024}, which may stabilize the quasi-Lindblad dynamics, even under the CP violation, but in the worst case, the quasi-Lindblad dynamics show the instability.}. We note that the hierarchical equations of motion approach~\cite{Tanimura1989, XuYanShietal2022} may also encounter similar stability challenges~\cite{DunnTempelaarReichman2019, Yan2020heomstability, KrugStockburger2023}. 

In the context of quantum simulation, \rev{violating} the CP condition \rev{implies} that the dynamics can no longer be efficiently \footnote{To encode the evolution using a quantum channel, an exponentially large subnormalization factor in $T$ must be introduced, rendering long-time simulation on a quantum computer prohibitively expensive.}  implemented as a quantum channel. \rev{In contrast,} the coupled Lindblad dynamics are inherently \rev{CPTP and thus compatible} with quantum hardware. The coupling between bosonic modes can also be realized in analog ion-trap-based quantum simulators~\cite{Lemmer2018, Chen2023, Katz2023,Kang2024}.

\vspace{1em}
\emph{How do coupled modes achieve efficiency and physicality?--} 
We first explain the physical intuition for the improved efficiency of coupled modes in \cref{thm:main}. While the coupled modes offer greater expressive power via $\Or(N^2)$ parameters compared to the $\Or(N)$ parameters of the decoupled pseudomode, it is not clear \emph{a priori} whether it translates into a scaling advantage. Furthermore, the higher dimensionality of the parameter space makes the search for optimal parameters more computationally demanding.

We resolve such difficulties by establishing a connection between the coupled Lindblad and quasi-Lindblad pseudomode theory, the latter of which is known to achieve the $\mathrm{polylog}(T/\varepsilon)$ scaling. This connection is characterized by the feasibility condition in Eq.~\eqref{eq:feasibility} of \cref{thm:main}, obtained from the Hermiticity and complete positivity of the pseudomode dynamics. The feasibility condition also naturally leads to a robust numerical algorithm for constructing the pseudomodes using semidefinite programming (Eq.~\eqref{eq:feasibility_sdp}).

\vspace{1em}
\emph{Theoretical analysis for the $\mathrm{polylog}  (T/\varepsilon)$ scaling and robust construction algorithms.--}

We first discuss the connection between the coupled Lindblad BCF $C^{\text c}(t) = g^\dagger \mathrm e^{-\mathrm iKt}g $  and the quasi-Lindblad BCF $C^{\text q}(t) = l^\dagger\mathrm e^{-\mathrm i\Lambda t}r$.
\rev{This is achieved by introducing a gauge transformation that generates couplings between different modes, while enforcing Hermiticity and positivity.}
The result is summarized in the following theorem. 

\begin{theorem}
Let $\hat{\rho}_{\text{S}}^{\text{c}}(t)$ and $\hat{\rho}_{\text{S}}^{\text{q}}(t)$ denote the reduced system density operators obtained from the coupled Lindblad and quasi-Lindblad theory, respectively. If the BCF coincide, then the reduced dynamics are identical:
\begin{equation}
    \begin{aligned}
      & C(t)  = C^{\text{c}}(t) = C^{\text{q}}(t) \Rightarrow \hat\rho_{\text{S}}(t) = \hat\rho_{\text{S}}^{\text{c}}(t) = \hat\rho_{\text{S}}^{\text{q}}(t).
    \end{aligned}
\end{equation}
Furthermore, if the following feasibility condition holds,

\rev{\begin{equation}
\begin{aligned}
\exists\, Y \succ 0 \ \text{s.t.} \quad
& Y r = l , \qquad && \text{(Hermiticity)} \\
& \mathrm{i}\!\left( Y \Lambda - \Lambda^{\dagger} Y \right) \succeq 0 ,
\qquad && \text{(Positivity)}
\end{aligned}
\label{eq:feasibility}
\end{equation}}
then there exists a coupled Lindblad BCF $C^{\text{c}}(t)$, with the same number of modes $N$ as the quasi-Lindblad pseudomode with its BCF $C^{\text{q}}(t)$, such that $C^{\text{c}}(t) = C^{\text{q}}(t)$.
\label{thm:main}
\end{theorem}

\begin{proof}
The first part of the theorem parallels the result of Ref.~\cite{TamascelliSmirneHuelgaetal2018}, which establishes that, for fixed $\hat{H}_{\text{S}}$, $\hat{S}$, and $\hat{\rho}_{\text{S}}(0)$, the BCF uniquely determines the reduced system dynamics.
 Therefore, we focus on the second part. A gauge transformation, $\Lambda\rightarrow K = X\Lambda X^{-1}$, $l^\dagger\rightarrow l^\dagger X^{-1} $, $r\rightarrow Xr$, with an invertible matrix $X$, leaves $C^{\text q}(t)$ invariant. The gauge-transformed BCF takes the coupled Lindblad form if the following conditions hold: (a) $g = (l^\dagger X^{-1})^\dagger =Xr$, (b) $\Gamma= (K^\dagger-K)/2\mathrm i\succeq 0$. 
These two conditions correspond to the Hermiticity and the positivity of the dynamics, respectively.
Introducing $Y = X^\dagger X$, we rewrite these as the equality and inequality constraints in \cref{eq:feasibility} by multiplying $X$ and $X^{\dag}$ to conditions (a) and (b). In addition, $X$ being invertible indicates that $Y \succ 0$, which makes conditions (a) and (b) equivalent to \cref{eq:feasibility}.
\end{proof}

\cref{thm:main} implies that when the feasibility condition is satisfied, the number of coupled Lindblad pseudomodes is \emph{no greater than} that of quasi-Lindblad pseudomodes. Consequently, the $\mathrm{polylog}(T/\varepsilon)$ scaling of the latter extends directly to the coupled Lindblad pseudomode.

In practice, the feasibility condition may not hold exactly. Nonetheless, we construct a numerical procedure that minimally violates it by solving the following least-squares problem with semidefinite constraints:
\begin{equation}
    \begin{aligned}
        & \underset{Y\succ 0}{\text{min}}\;
        \|l - Y r \|_2^2,\quad \text{ subject to }\quad  \mathrm i (Y \Lambda - \Lambda^\dagger Y)\succeq 0,
    \end{aligned}
    \label{eq:feasibility_sdp}
\end{equation}
This problem could be solved efficiently via a semidefinite programming (SDP) solver. Setting $X = \sqrt{Y}$, we recover the parameters $g$, $H$, and $\Gamma$ in Eq.~\eqref{eq:matching_Lindblad}.  This approach avoids the non-convex optimization used in prior works~\cite{LednevGarciaFeist2024, MascherpaSmirneSomozaEtAl2020, Dorda2014, Dorda2015}. We also propose a new procedure for obtaining coupled Lindblad parameters directly from the BCF in the frequency domain, motivated by the \emph{realization problem} in control theory~\cite{Antoulas1Anderson1986,MayoAntoulas2007},  as explained in the Supplemental Material (SM) \cite[Section S1]{SuppInfo}.
% \cref{sec:realization}
\vspace{1em}
\emph{General applicability.}-- The above discussion focuses on the single-site spin-boson model. Our coupled Lindblad pseudomode theory can be directly extended to multi-site systems as well as to fermionic environments. We refer to the SM \cite{SuppInfo}, Section S2 and Section S3, for details.
% \cref{sec:multi}  \cref{sec:fermion}
\vspace{1em}

\begin{figure}[ht]
  \centering
  \includegraphics[width=\linewidth]{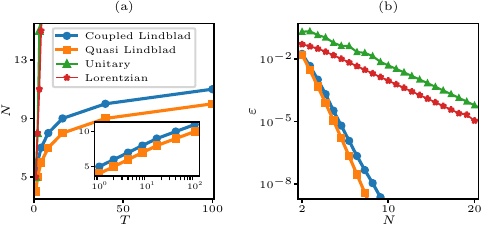}
  \caption{(a) For a fixed precision $\varepsilon = 10^{-6}$ in fitting $C(t)$ for $t \in [0, T]$, we plot the number of modes required, $N$, against the maximum simulation time $T$. The number of coupled Lindblad and quasi-Lindblad pseudomodes scales as $\Or(\log T)$ in contrast to the $\Or(T)$ scaling in the unitary and Lorentzian modes. 
(b) For a fixed $T = 10$, we plot $\varepsilon$ versus $N$, where the coupled Lindblad and quasi-Lindblad methods  achieve a significantly faster convergence rate.
}
\label{fig:comparison}
\end{figure}

\emph{Numerical results.--}
\rev{We begin by approximating the BCF $C(t)$ with coupled Lindblad modes to demonstrate their $\mathrm{polylog}(T/\varepsilon)$ scaling in comparison with other pseudomode approaches.} Our target is to fit $C(t)$ for $t \in [0, T]$, derived from the Ohmic spectral density $J(\omega) = \omega\mathrm e^{-\omega / \omega_c}$ for $\omega\geq 0$ at zero temperature and $\omega_c=1$. For the coupled Lindblad fitting, we first fit $\widetilde{C}(\omega)$ in the frequency domain using the realization-based method, and further refine the result by using it as an initial guess for a gradient-based optimization of $C(t)$ in the time domain. We illustrate $T$-dependence (Fig.~\ref{fig:comparison}(a)) and $\varepsilon$-dependence (Fig.~\ref{fig:comparison}(b)) of the number of modes $N$, where we use the averaged $L^2$ error $\varepsilon$ \cite[Eq. (S9)]{SuppInfo} as a measure of precision. \rev{See SM \cite[Section S6]{SuppInfo} for more details on the numerical setup.} %\cref{eq:average_L2} %\cref{sec:BCF_fitting}
The results of Fig.~\ref{fig:comparison}(a) with a fixed precision $\varepsilon=10^{-6}$ confirm that the coupled Lindblad pseudomode exhibits $N = \Or(\log T)$ scaling, similar to the quasi-Lindblad pseudomode~\cite{ThoennissVilkoviskiyAbanin2025}, in contrast to $N = \Or(T)$ scaling in the unitary and Lorentzian modes. In Fig.~\ref{fig:comparison}(b), with a fixed $T=10$, all methods exhibit $N= \Or(\log(1/\varepsilon))$ scaling, but the coupled Lindblad and quasi-Lindblad approaches converge significantly faster. The performance of the coupled Lindblad pseudomode closely matches that of the quasi-Lindblad pseudomode, indicating that the violation of the feasibility condition is small. \rev{Furthermore,
we confirm that the $\mathrm{polylog}(T/\varepsilon)$ efficiency of the coupled pseudomodes remains valid across additional numerical examples, specifically considering the BCF from sub-Ohmic and semicircular spectral densities (see SM, \cite[ Fig. S1 and Fig. S3]{SuppInfo}).} %\cref{fig:fermion} \cref{fig:Ohmic}

\begin{figure}[ht]
  \centering
  \includegraphics{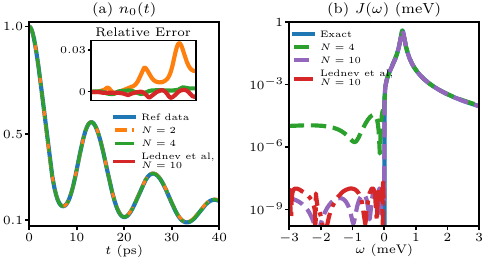}
  \caption{(a) Population $n_0(t)$ and its relative error for the spin-boson model dynamics. (b) Spectral density $J(\omega)$, and its fitting using coupled modes $N=4$ and $N=10$. Both plots are compared with results for $N=10$ extracted from Ref.~\cite{LednevGarciaFeist2024}. 
  }
  \label{fig:dynamics}
\end{figure}

\rev{Next, we demonstrate that the coupled Lindblad pseudomode theory can accurately capture real-time system dynamics.} We study the spin-boson model with $\hat{H}_{\text{S}}= \frac{\omega_e}{2} \hat{\sigma}_z $ and $\hat{S} = \hat{\sigma}_x$, where $\hat{\sigma}_{x/z}$ are Pauli operators. We follow the setup in Ref.~\cite{LednevGarciaFeist2024} by choosing the Lorentzian-like spectral density, $J(\omega)=\frac{2g^2 \kappa \omega_c \omega / \pi}{(\omega_c^2 - \omega^2)^2 + \kappa^2 \omega^2}$ for $\omega\geq 0$  and $J(\omega) = 0$ for $\omega < 0$. We focus on the ultra-strong coupling regime, with parameters $\omega_c = \omega_e = 0.58, \ g=0.25,$ and $\kappa = 0.1 \text{meV}$ \cite{LednevGarciaFeist2024}. \cref{fig:dynamics}(a) describes population dynamics $n_0(t)= \langle 0 |\hat{\rho}_S(t)|0 \rangle$ evolved from the initial state $\hat{\rho}_S(0)=|0\rangle\! \langle 0|$. We use $N=2$ and $N=4$ coupled pseudomodes, compared to the reference data obtained from unitary dynamics with a large discretization $N=400$. \rev{We use a time-dependent density matrix renormalization group algorithm for numerical simulations (see SM \cite[Section S4]{SuppInfo} for details).} %\cref{sec:dmrg}
Our approach achieves accuracy comparable to the coupled pseudomode dynamics in Ref.~\cite{LednevGarciaFeist2024} with only $N=4$ modes, versus their $N=10$. Notably, while Ref.~\cite{LednevGarciaFeist2024} requires non-convex optimization and explicit penalties to suppress unphysical negative-frequency contributions, our construction inherently minimizes these contributions. In Fig.~\ref{fig:dynamics}(b), the fitted BCF yields a minimal negative-frequency contribution of approximately $10^{-5}$ for $N=4$, which further vanishes to $10^{-9}$ for $N=10$. We also tested the method for the fermionic Anderson impurity model (see SM \cite[Section S5]{SuppInfo}), demonstrating its applicability to fermionic environments.%\cref{sec:fermion_test}

\begin{figure}[ht]
  \centering
  \includegraphics{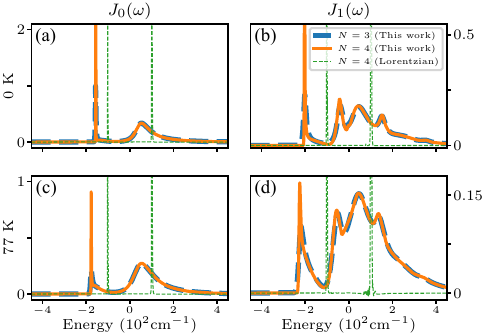}
  \caption{Normalized absorption spectrum $S(\omega)$ for the dimer model with two different environments  $J_0(\omega)$ \rev{(a, c)} and $J_1(\omega)$ \rev{(b, d)}, at zero \rev{(a, b)} and finite temperature (77K, \rev{(c, d)}). 
}\label{fig:absorption_spec}
\end{figure}

Finally, we compute the absorption spectra using the coupled Lindblad pseudomode for a dimer model with three states, $\ket{g}, \ket{\epsilon_1},$ and $\ket{\epsilon_2}$, and $\hat{H}_{\text{S}} = \sum_{i=1}^2 \epsilon_i \ket{\epsilon_i}\!\bra{\epsilon_i} + J (\ket{\epsilon_1}\!\bra{\epsilon_2}+ \ket{\epsilon_2}\!\bra{\epsilon_1}) $. The environment is coupled independently to each excited state via the operators, $\hat{S}_i = \ket{\epsilon_i}\!\bra{\epsilon_i}$, using the parameters of Ref.~\cite{MascherpaSmirneSomozaEtAl2020}. We consider two different spectral densities: $J_0(\omega)$, a broad spectrum, and $J_1(\omega)$, an additional sharp peak over $J_0(\omega)$ \rev{(details in SM \cite[Section S7]{SuppInfo})}. We compute the absorption spectrum, $S(\omega) =\omega \operatorname{Im}\left(\int_0^{\infty}\mathrm i C_{\hat{\mu}}(t) \mathrm e^{\mathrm i\omega t} \mathrm dt\right)$, derived from the dipole-dipole correlation function $C_{\hat{\mu}}(t)$ with $\hat{\mu} = \sum_i \ket{\epsilon_i}\!\bra{g} + \ket{g}\!\bra{\epsilon_i} $ and an initial state $\hat{\rho}_{\text{S}}(0) = \ket{g}\!\bra{g}$. The result is shown in Fig.~\ref{fig:absorption_spec} compared to results from the Lorentzian pseudomode. Notably, the absorption spectrum exhibits a sharp peak in the negative frequency region at zero temperature, which becomes narrower and sharper as $N$ increases, highlighting the convergence behavior of our approach. Accurately capturing the presence of sharp resonances and broad features in the spectrum requires a faithful reconstruction of the spectral density, which our method achieves effectively. In contrast, the Lorentzian pseudomode fails to reproduce the broad component and misplaces several spectral peaks. %\cref{sec:absorption}

\vspace{1em}
\emph{Conclusions and outlook.--} We argue that the coupled Lindblad pseudomode framework possesses all the desirable features of a well-designed pseudomode theory.
 The dynamics can be realized as a quantum channel, making them inherently stable. In both theory and practice, only a small number of pseudomodes is needed to achieve accurate results with $\text{polylog}(T/\varepsilon)$ scaling. Moreover, the pseudomodes can be constructed via a robust algorithm.
We illustrate the method using the spin-boson model as an example, but it is equally applicable to fermionic environments. We anticipate that with the improved understanding and robustness provided by the techniques developed in this work, the coupled Lindblad framework will prove broadly useful across a wide range of open quantum system applications. These include the simulation of ultra-strong coupling regimes~\cite{RMP_FornDiaz2019}, the use of pseudomodes as an impurity solver in dynamical mean-field theory~\cite{KotliarSavrasovHauleetal2006, AokiTsujiEcksteinetal2014, BertrandBesserveFerreroetal2024}, and the simulation of condensed-phase chemical dynamics on quantum platforms such as trapped-ion devices~\cite{Lemmer2018, Kang2024}.

\vspace{1em}
\emph{Acknowledgements.--} This material is based upon work supported by the U.S. Department of Energy, Office of Science, Accelerated Research in Quantum Computing Centers, Quantum Utility through Advanced Computational Quantum Algorithms, grant no. DE-SC0025572 (Z.H., G.K.C., L.L.). Additional support is acknowledged from the Simons Targeted Grant in Mathematics and Physical Sciences on Moir\'e Materials Magic (Z.H., L.L.) and the U.S. Department of Energy, Office of Science, Office of Advanced Scientific Computing Research and Office of Basic Energy Sciences, Scientific Discovery through Advanced Computing (SciDAC) program under Award Number DE-SC0022088 (G.P., G.K.C.). G.P. acknowledges support from the Eddleman Quantum Graduate Fellowship at Caltech. G.K.C. and L.L. are Simons Investigators. Additional support for G.K.C. in the early phase of this project was provided by the U.S. Department of Energy, Office of Science, Basic Energy Sciences, via Award Number DE-SC0019374. This research used the Savio computational cluster resource provided by the Berkeley Research Computing program at the University of California, Berkeley. We thank helpful discussions with Zhiyan Ding, Mingyu Kang, David Limmer, Olivier Parcollet, Miles Stoudenmire, Steve White, Jason Kaye, and Yuanran Zhu. 

\bibliographystyle{apsrev4-2}
\bibliography{ref}% Produces the bibliography via BibTeX.
\newpage 
\clearpage
\thispagestyle{empty}
\onecolumngrid
\begin{center}
\textbf{\large Supplemental Material for \\ Coupled Lindblad pseudomode theory for simulating open quantum systems}
\end{center}
 
\begin{CJK*}{UTF8}{mj}
\begin{center}
Zhen Huang,$^1$ Gunhee Park (박건희),$^{2}$, Garnet Kin-Lic Chan$^3$ and Lin Lin$^{1,4}$\\
\smallskip
\small{\emph{$^1$\DeptMath\\$^2$\Caltechappliedscience\\$^3$\CaltechChem\\$^4$\LBLMath}}\\
(Dated: \today)

\end{center}
\end{CJK*}
\setcounter{equation}{0}
\setcounter{figure}{0}
\setcounter{table}{0}
\setcounter{page}{1}
\setcounter{section}{0}
\setcounter{secnumdepth}{3}
\makeatletter
\renewcommand{\theequation}{S\arabic{equation}}
\renewcommand{\thefigure}{S\arabic{figure}}
\renewcommand{\thesection}{S\arabic{section}}
\renewcommand{\bibnumfmt}[1]{[S#1]}
%\renewcommand{\citenumfont}[1]{S#1} % Prefix a "S" to all equations, figures, tables and reset the counter 
% \author{Zhen Huang}
% \affiliation{Department of Mathematics, University of California, Berkeley, California 94720, USA}
% \author{Gunhee Park}
% %\author{Gunhee Park (박건희)}
% \affiliation{Division of Engineering and Applied Science, California Institute of Technology, Pasadena, California 91125, USA}

% \author{Garnet Kin-Lic Chan}
% \affiliation{Division of Chemistry and Chemical Engineering, California Institute of Technology, Pasadena, California 91125, USA}
% \author{Lin Lin}
% \thanks{linlin@math.berkeley.edu}
% \affiliation{Department of Mathematics, University of California, Berkeley, California 94720, USA}
% \affiliation{Applied Mathematics and Computational Research Division, Lawrence Berkeley National Laboratory, Berkeley, California 94720, USA}

% \maketitle

\onecolumngrid

\section{Realization-based coupled Lindblad modes construction}
\label{sec:realization}
In this section, we describe a new procedure for obtaining coupled Lindblad parameters directly from the BCF in the frequency domain, motivated by the \emph{realization problem} in control theory~\cite{Antoulas1Anderson1986, MayoAntoulas2007}. It is worth noting that the fitting of $C^\text{q}(t)$ in \cref{eq:matching_quasi} is often performed in the time domain, using signal processing algorithms such as ESPRIT~\cite{RoyKailath1989, ParkHuangZhuetal2024} and Prony algorithm \cite{ChenWangZhengetal2022}. However, if the bath information is provided in the frequency domain rather than in the time domain, a Fourier transform of the spectral density is required, which can introduce additional approximation errors, particularly when \( J(\omega) \) is only available on a discrete frequency grid with limited accuracy.

The realization problem aims to fit a given scalar function \( f \) in the form \( f(\omega) = l^\dagger (K - \omega I)^{-1} r \). The algorithm is highly robust and only requires applying a singular value decomposition (SVD) to the \emph{Loewner matrix}~\cite{MayoAntoulas2007} constructed from the sampled data \( f(\omega) \).

In our setting, given \( \widetilde{C}(\omega) \) sampled on a frequency grid, we seek parameters satisfying 
\begin{equation}
\widetilde{C}(\omega) \approx \mathrm{Im}(g^\dagger (K - \omega I)^{-1}g), \quad K = H - \mathrm{i} \Gamma. 
\label{eqn:matching_freq}
\end{equation}
This introduces two key differences from the conventional realization problem: (1) we only have access to the imaginary part of the meromorphic function; (2) we need to enforce the constraints \( l = r \) and \( \Gamma = (K^\dagger - K)/2\mathrm{i} \succeq 0 \).

To address these challenges, we first observe that
$\widetilde{C}(\omega) = \mathrm{Im}(g^\dagger (K - \omega I)^{-1}g) = \frac{1}{2 \mathrm i}  \mathsf{g}_{l}^\dagger (\mathsf{K} - \omega I)^{-1} \mathsf{g}_r$, where \( \mathsf{g}_l = \begin{pmatrix} g \\ \overline{g} \end{pmatrix} \), \( \mathsf{g}_r = \begin{pmatrix} g \\ -\overline{g} \end{pmatrix} \), and \( \mathsf{K} = \operatorname{diag}(K, \overline{K}) \). Thus, by fitting \( \widetilde{C}(\omega) \) within the standard realization framework, we obtain an approximation of the form $\widetilde C(\omega) = \frac{1}{2\pi \mathrm{i}} \widetilde{l}^\dagger (\widetilde{K} - \omega I)^{-1} \widetilde{r}$, where \( \widetilde{l} \), \( \widetilde{r} \), and \( \widetilde{K} \) are related to $\mathsf g_l$, $\mathsf g_r$ and $\mathsf K$ via an undetermined gauge. We then determine this gauge by enforcing the physical constraints, leading to an SDP problem similar to Eq.~\eqref{eq:feasibility}, which can be solved efficiently via a robust SDP subroutine.

\section{Coupled Lindblad pseudomodes theory for multi-site cases}
\label{sec:multi} 
The generalization from a single-site to a multi-site spin-boson system  is straightforward. The key difference is in the system-bath coupling term $\hat H_{\text{SA}}$, which is characterized by a coupling coefficient matrix $\textbf g$ of size $ N\times n$:
\begin{equation}
    \hat H_{\text{SA}} = \sum_{j=1}^n \hat S_j\hat A_j,\quad \hat A_j = \sum_{k=1}^N \textbf g_{kj} \hat b_k + \overline{\textbf g_{kj}}\hat b_k^\dagger.
    \label{eq:H_SA_multi}
\end{equation}
Here $N$ is the number of pseudomodes and $n$ is the number of terms in $\hat H_{\text{SA}}$. In other words, the generalization falls upon replacing the coupling vector $g$ with a matrix $\textbf g$. The corresponding BCF, which is a $n\times n$ matrix-valued function, takes the following form: 
\begin{equation}
    C^{\text{c}}(t) =\textbf g^\dagger \mathrm e^{-\mathrm iKt}\textbf g,\quad 
    \textbf g\in\mathbb C^{N\times n} , \quad K = H-\mathrm i\Gamma,\quad H = H^\dagger,\quad\Gamma\succeq 0.
\end{equation}
Therefore, the coupled Lindblad dynamics takes the following form:
\begin{equation}
\frac{\mathrm d\hat\rho}{\mathrm dt}= -\mathrm i [\hat H_{\text S}+\hat H_{\text A}+\hat H_{\text{SA}},\hat\rho] + \boldsymbol{D}_{\text A}(\hat\rho),
\end{equation}
where $\hat H_{\text S}$ is the system Hamiltonian, $\hat H_{\text{SA}}$ is the system-bath coupling Hamiltonian as in \cref{eq:H_SA_multi}, and $\hat H_{\text A}$, $\boldsymbol{D}_{\text A}$ are the bath Hamiltonian and dissipation as in \cref{eq:coupled_Lindblad}.
\section{Coupled Lindblad pseudomode theory for fermionic systems }
\label{sec:fermion}
We recall that in the bosonic case, the system-bath coupling takes the form $\hat H_{\text{SB}} = \hat S\hat B$, where $\hat S$ is Hermitian and $B = \int \sqrt{J(\omega)}(\hat b_\omega + \hat b_\omega^\dagger)$ is also Hermitian. As a result,  only one bath correlation function is required to capture the bath's influence on the system. However, in general, when $S$ is non-Hermitian and thus the system-bath coupling takes the form $\hat H_{\text{SB}} = \hat S\hat B + \hat B^\dagger\hat S^\dagger$, then two different BCFs $\langle B(t) B^\dagger\rangle$ and $\langle B^\dagger(t)B \rangle$ will arise. This is precisely the case for the fermionic impurity problems, in which we need to consider both the lesser and the greater BCFs.

For fermionic cases, let us consider an impurity Hamiltonian $\hat H_{\text S}(\hat a_i,\hat a_i^\dagger)$ coupled to a bath with chemical potential $\mu$ via the following system-bath coupling $\hat H_{\text{SB}}$:
\begin{equation}
    \hat H_{\text{SB}} = \sum_{i=1}^n\int\mathrm d\omega  f_i(\omega) \hat a_i^\dagger \hat c_{\omega} + \text{h.c.},
    \label{eq:fermion_SB}
\end{equation}
The bath's influence on the system is characterized by the lesser and greater hybridization functions ($\Delta^<(t)$ and $\Delta^>(t)$):

\begin{equation}
    \Delta^<(t) = \int J(\omega) f_{\text{FD}}^{\mu,\beta}(\omega)\mathrm e^{-\mathrm i\omega t}\mathrm d\omega, \quad  \Delta^>(t) = \int J(\omega) (1 - f_{\text{FD}}^{\mu,\beta}(\omega))\mathrm e^{-\mathrm i\omega t}\mathrm d\omega. 
    \label{eq:real_hyb}
\end{equation}
Here $J(\omega)=f_i(\omega)\overline{f_j(\omega)}$ is the bath spectral density, $\mu$ is the chemical potential, $\beta$ is the inverse temperature, and $f_{\text{FD}}^{\mu,\beta}(\omega) = \frac{1}{1+\mathrm e^{\beta(\omega-\mu)}}$ is the Fermi-Dirac function. To apply the pseudomode theory for fermionic cases, the key is to account for both lesser  and greater real-time hybridization functions  arising in the problem. This is similar to what is done in quasi-Lindblad theories \cite[Sec II. D]{ParkHuangZhuetal2024}. We conduct fitting for both $\Delta^<(t)$ and $\Delta^>(t)$ as follows:
\begin{equation}
    \Delta^<(t)\approx \Delta^{\text c,<}(t) = (\textbf g^<)^\dagger \mathrm e^{(-\mathrm iH^<-\Gamma^<)t}\textbf g^<,\quad 
   \Delta^>(t)\approx \Delta^{\text c,>}(t) = (\textbf g^>)^\dagger \mathrm e^{(-\mathrm iH^>-\Gamma^>)t}\textbf g^>.
   \label{eq:hyb_fitting}
\end{equation}
Here $\textbf g^<$, $H^<$, $\Gamma^<$ are of size $(n,N_1)$, $(N_1,N_1)$ and $(N_1,N_1)$ and $\textbf g^>$, $H^>$, $\Gamma^>$ are of size $(n,N_2)$, $(N_2,N_2)$ and $(N_2,N_2)$. With the hybridization fitting \cref{eq:hyb_fitting}, the coupled Lindblad dynamics is as follows:
\begin{equation}
\begin{aligned}
   \frac{\mathrm d}{\mathrm dt}\hat\rho_{\text{SA}} & = 
    -\mathrm i[\hat H_{\text S},\hat\rho]  -\mathrm i[\hat H_{\text A_1} + \hat H_{\text{SA}_1},\hat\rho] -\mathrm i[ \hat H_{\text A_2} + \hat H_{\text{SA}_2},\hat\rho]+ 
    \boldsymbol{D}_{\text{A}_1}(\hat\rho) + \boldsymbol{D}_{\text{A}_2}(\hat\rho), \quad \\
   & \hat H_{\text A_1}  =\sum_{k,l=1}^{N_1} H^<_{kl} \hat c_k^\dagger \hat c_l,\quad \hat H_{\text{SA}_1} = \sum_{i=1}^n \sum_{k=1}^{N_1} (\textbf g^<)_{ki} \hat c_k^\dagger \hat a_i+\text{h.c.},
   \\
& \hat H_{\text A_2}  =\sum_{k,l=1}^{N_2} H^>_{kl} \hat d_k^\dagger \hat d_l,\quad
    \hat H_{\text{SA}_2} = \sum_{i=1}^n \sum_{k=1}^{N_2} (\textbf g^>)_{ki} \hat d_k^\dagger \hat a_i +\text{h.c.} 
    ,\\ & \boldsymbol{D}_{\text{A}_1}(\hat\rho) = \sum_{k,l=1}^{N_2} \Gamma^<_{kl} (2\hat c_k^\dagger\hat\rho\hat c_l - \{\hat c_l \hat c_k^\dagger,\hat\rho\} ),\quad \boldsymbol{D}_{\text{A}_2}(\hat\rho) = \sum_{k,l=1}^{N_2} \Gamma^>_{kl} (2\hat d_l\hat\rho\hat d_k^\dagger - \{\hat d_k^\dagger \hat d_l,\hat\rho\} ).\\
    &\hat\rho(0) = \hat\rho_{\text S}(0)\otimes\bigotimes_{k=1}^{N_1}|1\rangle\langle 1|\otimes \bigotimes_{l=1}^{N_2}|0\rangle\langle 0|.
\end{aligned}
\label{eq:fermion_coupled_lind}
\end{equation}
We refer to \cite{HuangLinParketal2024} for details of the proof of correctness for the coupled Lindblad dynamics \cref{eq:fermion_coupled_lind}.

\section{TD-DMRG based simulation}
\label{sec:dmrg}
To solve the coupled Lindblad dynamics, we first rewrite the dynamics of the density operator in the superoperator formalism. We use the time-dependent density matrix renormalization group (TD-DMRG) method to evolve the density operators, which are propagated using the time-dependent variational principle (TDVP), implemented in Julia packages \texttt{ITensors.jl} and \texttt{ITensorMPS.jl} \cite{ITensor, ITensor-r0.3}. As for the ordering of sites, we follow \cite{ParkHuangZhuetal2024}, in which we order the sites based on the magnitude of the dissipation. We set the cutoff threshold $\epsilon=10^{-12}$. After each TDVP step, we normalize the state to have a trace 1, i.e., $\operatorname{tr}(\hat\rho)=1$. To evaluate any physical observable $\hat O$, we calculate the trace $\operatorname{tr}(\hat O\hat\rho)$. We remark that to take traces of any operator $\hat A$, in the superoperator formalism, it means to calculate the inner product $\langle\langle I|A\rangle\rangle$, where $|I\rangle\rangle$ is the vectorization of the identity operator, and which can be explicitly constructed as a matrix product state.

For benchmarking purposes, we calculate the reference system density dynamics for both the spin-boson model and for the fermionic impurity model presented in the next section. This is enabled by a standard efficient unitary discretization of the bath using Gaussian quadrature and Legendre polynomials \cite{DeVegaInesSchollwock2015}, with $N = 400$ orbitals in the former case and $N=200$ spin-orbitals in the latter case. 

\section{Numerical experiments on the fermionic Anderson impurity model}
\label{sec:fermion_test}
Here we numerically demonstrate the applicability of our theories to fermionic problems in \cref{fig:fermion}.

\begin{figure}[ht]
    \centering
    \includegraphics[width=0.65\linewidth]{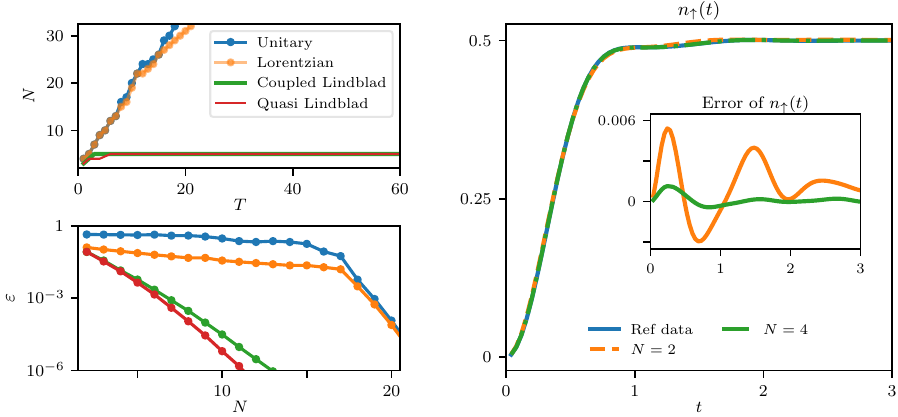}
    \caption{Numerical experiments on the Fermionic Anderson impurity model with semicircular bath spectral density using the coupled Lindblad approach. (Left) Results of BCF fitting. (Right) Dynamics of $n_\uparrow(t)$.}
    \label{fig:fermion}
\end{figure}

In \cref{fig:fermion} (left), similar to Fig. 1, we present the fitting of the BCF corresponding to the semicircular spectral density, $J(\omega) = \frac{\Gamma}{\pi}\sqrt{1 - \frac{\omega^2}{W^2}}$ with half bandwidth $W=10$, $\Gamma = 1$, and inverse temperature $\beta =100$. We consider a single-impurity Anderson model with impurity Hamiltonian $\hat H_{\text{S}} = \epsilon (\hat n_\uparrow+ \hat n_\downarrow) + U\hat n_\uparrow\hat n_\downarrow $ where $\hat{n}_{\uparrow / \downarrow} = \hat{a}^\dagger_{\uparrow / \downarrow}  \hat{a}_{\uparrow / \downarrow}$ with $U=8$ and  $ \epsilon = -4$.  The fermionic system-bath coupling is defined in \cref{eq:fermion_SB}, with the above-mentioned semicircular density $J(\omega)$. On the right, we present the time evolution of \( n_\uparrow(t)  = \langle  \hat{n}_\uparrow \hat{\rho}_S(t) \rangle \) with an initial empty impurity $\hat \rho_{\text{S}}(0) = |0\rangle\langle 0|$. Recall that as mentioned in \cref{eq:fermion_coupled_lind}, $N_1$, $N_2$ are the number of pseudomodes (per spin) for the lesser and greater BCFs, respectively. In this experiment, we take $N_1=N_2=N$. Remarkably, the coupled Lindblad pseudomode framework achieves high accuracy with only a small number of pseudomodes (\( N = 2, 4 \)), as shown in \cref{fig:fermion} (Right).

\section{Additional information on benchmarking  BCF fitting}
\label{sec:BCF_fitting}
In the main text, we have benchmarked the coupled Lindblad approach against various other methods using the Ohmic spectral density in Fig. 1.
The fitting error is evaluated as follows:
\begin{equation}
    \varepsilon = \left(\frac 1 T\int_0^T \left|C(t) - C_{\text{approx}}(t)\right|^2\mathrm dt\right)^{1/2}
    \label{eq:average_L2}
\end{equation}
In addition, here we show the fitting results in the time domain \rev{for the Ohmic spectral density} using $N=4$ modes in \cref{fig:compare_Ct}.
\begin{figure}[ht]
    \centering
    \includegraphics[width=0.6\linewidth]{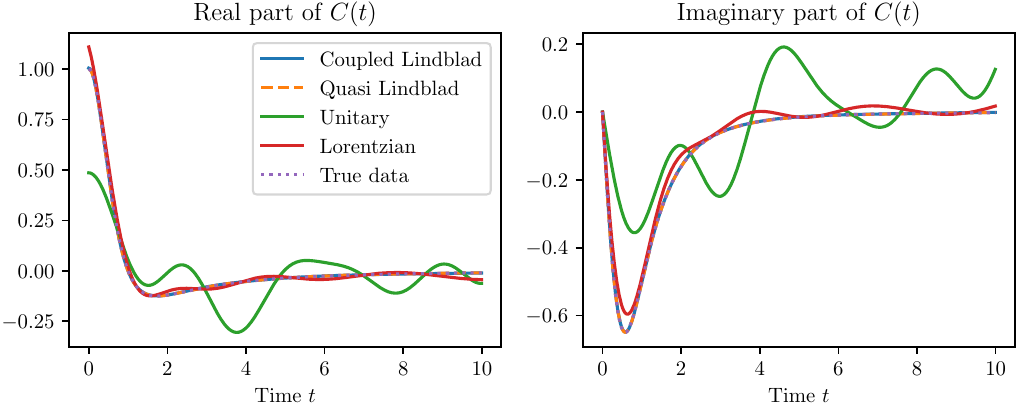}
    \caption{Comparison of unitary, Lorentzian, coupled Lindblad (this work) and quasi-Lindblad pseudomodes for fitting the same BCF using $N=4$ modes. }
    \label{fig:compare_Ct}
\end{figure}

\rev{To further demonstrate the expressibility of the coupled Lindblad modes, here we also show the fitting results for sub-Ohmic densities $J(\omega) = \omega^\alpha\mathrm e^{-\omega/\omega_c}$ with $\alpha=1/2$ and $\omega_c = 1.0$. Sub-Ohmic densities are known to be more challenging to simulate compared to Ohmic and super-Ohmic densitites  due to its singularity at $\omega=0$, and exhibits very rich dynamical phases (see \cite{goulko2025transient} for example). Our results in \cref{fig:Ohmic} exhibits the same $O(\mathrm{poly}\log(T/\varepsilon))$ scaling as in other examples, and achieve very accurate fitting with only $N=4$ bath modes.
}
\begin{figure}
    \includegraphics[width=100mm]{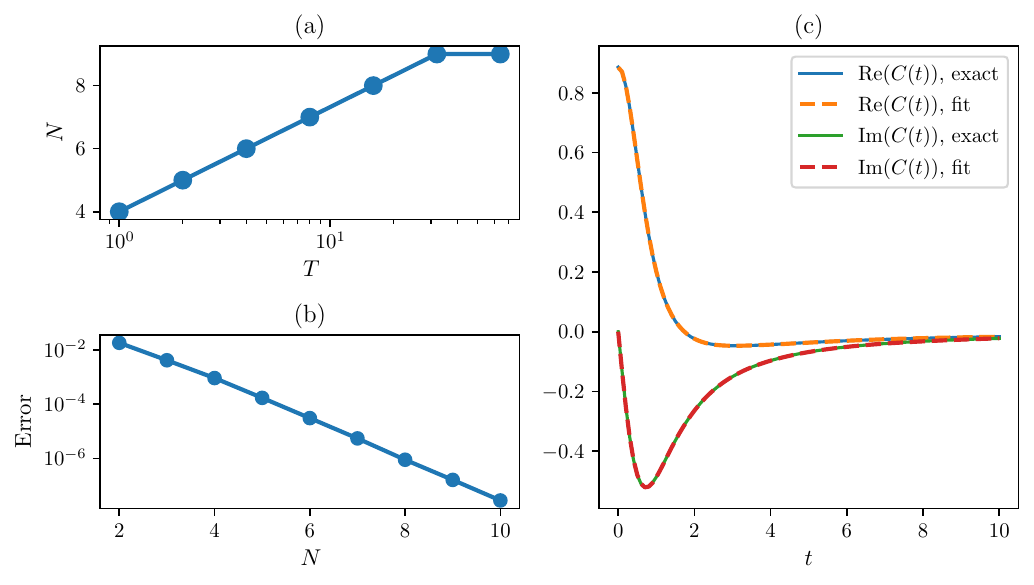}
    \caption{\rev{Results for sub-Ohmic density $J(\omega) = \omega^{1/2}\mathrm e^{-\omega}$ using coupled Lindblad modes. (a) Number of coupled modes $N$ v.s. maximum simulation time $T$ for reaching $\varepsilon=10^{-5}$ accuracy. (b) Error v.s. number of coupled modes, for fixed maximum simulation time $T=10$. (c) Fitting results of $C(t)$, with $N=4$ coupled modes.}}
    \label{fig:Ohmic}
\end{figure}

\section{Additional Information on calculating the absorption spectrum}
\label{sec:absorption}
Finally, we comment on the evaluation of real-time correlation functions. Following \cite{MascherpaSmirneSomozaEtAl2020}, the correlation function takes the form $C_{\hat{\mu}}(t) = \operatorname{tr}\left(\hat{\mu}^\dagger \mathrm e^{\mathcal Lt} \hat{\mu} \hat{\rho}_{\text S}(0) \otimes \hat\rho_{\text B}(0)\right)$,  where $\mathcal L$ is the Liouvillian superoperator corresponding to the coupled Lindblad dynamics, $\hat{\rho}_{\text S}(0) = |g \rangle \! \langle g |$, and $\hat{\mu} = \sum_i \ket{\epsilon_i}\!\bra{g} + \ket{g}\!\bra{\epsilon_i}$. The absorption spectrum is obtained via the Fourier transform of $C_{\hat{\mu}}(t)$, aided with ESPRIT as done in \cite{SheeHuangHeadGordon2025}. We calculate up to time $T=5000$ with time step $\Delta t = 0.0005$.

In this example, we use two  external environments, $J_0(\omega)$ and $J_1(\omega)$, with distinct features (see \cref{fig:j0_j1}), adapted from \cite{MascherpaSmirneSomozaEtAl2020}. $J_0(\omega)$ features broad spectrum, originally proposed in \cite{ADOLPHS20062778}, known as the Adolphs-Renger form. $J_1(\omega) = J_0(\omega) + J_{\text{AL}}(\omega)$ has an additional anti-symmetrized Lorentzian peak, 
\begin{equation}
    J_{\text{AL}}(\omega) = S \frac{8 \Gamma \Omega (4\Omega^2 + \Gamma^2)\omega}{(4(\omega - \Omega)^2+\Gamma^2) (4(\omega + \Omega)^2+\Gamma^2)}. 
\end{equation}
The parameters in this form are taken from \cite{MascherpaSmirneSomozaEtAl2020}.

The goal is to verify that our coupled Lindblad pseudomode theory works well under different external environments.
Although the spectrum under these two environments varies significantly (see \cref{fig:absorption_spec}), both cases are well captured by three coupled modes as shown in \cref{fig:absorption_spec}.

\begin{figure}[ht]
    \centering
    \includegraphics[width=0.6\linewidth]{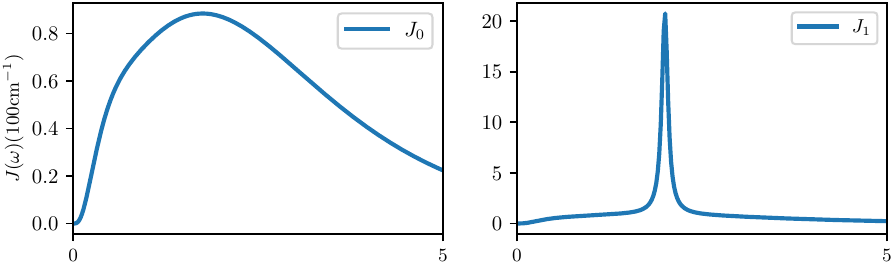}
    \caption{Spectral function of the bath environments $J_0(\omega)$ and $J_1(\omega)$.}
    \label{fig:j0_j1}
\end{figure}

\end{document}